\documentclass{amsart}
\usepackage[varg]{txfonts}
\usepackage[dvips]{graphicx,psfrag}
\usepackage[subrefformat=parens]{subcaption}
\usepackage{enumitem}
\usepackage[norelsize,linesnumbered,algoruled]{algorithm2e}
\usepackage[cm]{fullpage}
\usepackage{hyperref}
\usepackage{amsthm}
\newtheorem{theorem}{Theorem}
\newtheorem{lemma}[theorem]{Lemma}
\newtheorem{claim}[theorem]{\it Claim}
\newtheorem{corollary}[theorem]{Corollary}
\newtheorem{problem}{Problem}

\title{A vertex ordering characterization of simple-triangle graphs}
\author{Asahi Takaoka}
\address{
  Department of Information Systems Creation, 
  Kanagawa University, 
  Rokkakubashi 3-27-1 Kanagawa-ku, 
  Kanagawa, 221--8686, Japan 
}
\email{takaoka@jindai.jp}

\keywords{
Alternatly orientable graphs, 
Linear-interval orders, 
PI graphs, 
PI orders, 
Simple-triangle graphs, 
Vertex ordering characterization
}

\begin{document}
\begin{abstract}
Consider two horizontal lines in the plane. 
A pair of a point on the top line and 
an interval on the bottom line 
defines a triangle between two lines. 
The intersection graph of such triangles 
is called a simple-triangle graph. 
This paper shows a vertex ordering characterization of simple-triangle graphs 
as follows: a graph is a simple-triangle graph 
if and only if there is a linear ordering of the vertices 
that contains both an alternating orientation of the graph and 
a transitive orientation of the complement of the graph. 
\end{abstract}

\maketitle

\section{Introduction}
Let $L_1$ and $L_2$ be two horizontal lines in the plane with $L_1$ above $L_2$. 
A pair of a point on the top line $L_1$ and an interval on the bottom line $L_2$ 
defines a triangle between $L_1$ and $L_2$. 
The point on $L_1$ is called the \emph{apex} of the triangle, and 
the interval on $L_2$ is called the \emph{base} of the triangle. 
A \emph{simple-triangle graph} is the intersection graph of such triangles, 
that is, a simple undirected graph $G$ is called a simple-triangle graph 
if there is such a triangle for each vertex 
and two vertices are adjacent if and only if 
the corresponding triangles have a nonempty intersection. 
The set of triangles is called a \emph{representation} of $G$. 
See Figures~\ref{fig:example}\subref{fig:C-graph} 
and~\ref{fig:example}\subref{fig:C-representation} for example. 
Simple-triangle graphs are also known as 
\emph{PI graphs}~\cite{BLS99,COS08-ENDM,CK87-CN}, 
where \emph{PI} stands for \emph{Point-Interval}. 
Simple-triangle graphs were introduced 
as a generalization of both interval graphs and permutation graphs, 
and they form a proper subclass of trapezoid graphs~\cite{CK87-CN}. 
Although a lot of research has been done for interval graphs, 
for permutation graphs, and for trapezoid graphs 
(see~\cite{BLS99,Golumbic04,GT04,MM99,Spinrad03} for example), 
there are few results 
for simple-triangle graphs~\cite{BLR10-Order,COS08-ENDM,CK87-CN}. 
The polynomial-time recognition algorithm has been 
given~\cite{Mertzios15-SIAMDM,Takaoka16-arXiv}, 
but the complexity of the graph isomorphism problem 
still remain an open question~\cite{Takaoka15-IEICE,Uehara14-DMTCS},
which makes it interesting to study the structural characterizations
of this graph class. 

A \emph{vertex ordering} of a graph $G = (V, E)$ is a linear ordering 
$\sigma = v_1, v_2, \ldots, v_n$ of the vertex set $V$ of $G$. 
A \emph{vertex ordering characterization} of a graph class $\mathcal{G}$ 
is a characterization of the following type: 
a graph $G$ is in $\mathcal{G}$ if and only if $G$ has 
a vertex ordering fulfilling some properties. 
See~\cite{BLS99,CS15-JGT} for example of vertex ordering characterizations. 
This paper shows a vertex ordering characterization of simple-triangle graphs. 
More precisely, we characterize the apex orderings of simple-triangle graphs. 
Here, we call a vertex ordering $\sigma$ of a simple-triangle graph $G$ 
an \emph{apex ordering} if there is a representation of $G$ such that 
$\sigma$ coincides with the ordering of the apices of the triangles in the representation. 
See Figure~\ref{fig:example}\subref{fig:C-ordering} for example. 

The organization of this paper is as follows. 
Before describing the vertex ordering characterization, 
we show in Section~\ref{section:orders} a characterization of 
the linear-interval orders, 
the partial orders associated with simple-triangle graphs. 
The vertex ordering characterization of simple-triangle graphs
is shown in Section~\ref{section:apex}. 
We remark some open questions and related topics 
in Section~\ref{section:conclusion}. 

\begin{figure*}[t]
  \psfrag{L1}{$L_1$}
  \psfrag{L2}{$L_2$}
  \psfrag{a1}{$a_1$}
  \psfrag{a2}{$a_2$}
  \psfrag{b1}{$b_1$}
  \psfrag{b2}{$b_2$}
  \psfrag{c1}{$c_1$}
  \psfrag{c2}{$c_2$}
  \psfrag{c3}{$c_3$}
  \psfrag{c4}{$c_4$}
  \centering\subcaptionbox{A graph $G$. \label{fig:C-graph}}
  {\includegraphics[scale=0.6]{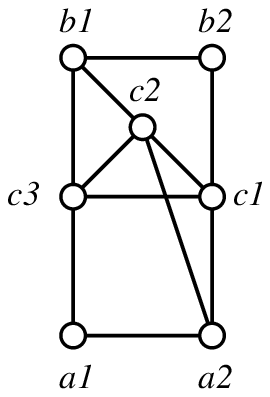}}
  \centering\subcaptionbox{The representation of $G$. \label{fig:C-representation}}
  {\includegraphics[scale=0.6]{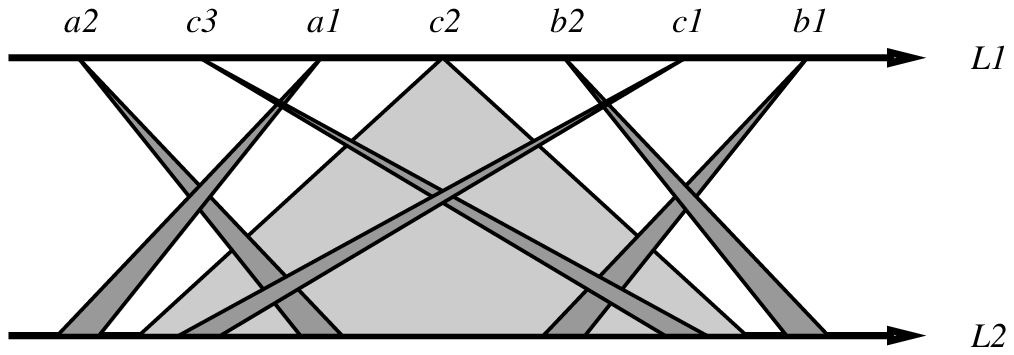}}
  \centering\subcaptionbox{The apex ordering of $G$. \label{fig:C-ordering}}
  {\includegraphics[scale=0.6]{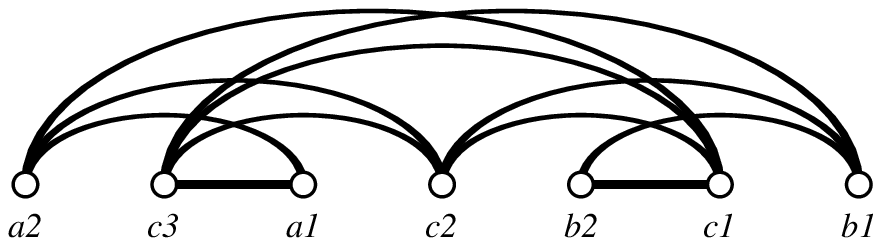}}
  \centering\subcaptionbox{The order $P$. \label{fig:C-order}}
  {\includegraphics[scale=0.6]{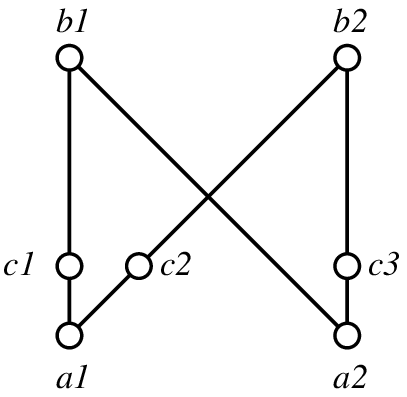}}
  \caption{
    A simple-triangle graph $G$, 
    the representation of $G$ consisting of the triangles, 
    the apex ordering of $G$, and 
    the Hasse diagram of linear-interval order $P$. 
    }
  \label{fig:example}
\end{figure*}

\section{Linear-interval orders}\label{section:orders}
A \emph{partial order} is a pair $P = (V, \prec_P)$, 
where $V$ is a finite set and $\prec_P$ is a binary relation on $V$ 
that is irreflexive and transitive. 
The finite set $V$ is called the \emph{ground set} of $P$. 
A partial order $P = (V, \prec_P)$ is called a \emph{linear order} 
if for any two elements $u, v \in V$, $u \prec_P v$ or $u \succ_P v$. 
A partial order $P = (V, \prec_P)$ is called an \emph{interval order} 
if for each element $v \in V$, there is an interval $I(v) = [l(v), r(v)]$ 
on the real line such that for any two elements $u, v \in V$, 
$u \prec_P v \iff r(u) < l(v)$, 
that is, $I(u)$ lies completely to the left of $I(v)$. 
The set of intervals $\{I(v) \mid v \in V\}$ is called 
an \emph{interval representation} of $P$. 

Let $P_1 = (V, \prec_1)$ and $P_2 = (V, \prec_2)$ be two partial orders 
with the same ground set. 
The \emph{intersection} of $P_1$ and $P_2$ is the partial order 
$P = (V, \prec_P)$ such that $u \prec_P v \iff u \prec_1 v$ and $u \prec_2 v$; 
it is denoted by $P = P_1 \cap P_2$. 
A partial order $P$ is called an \emph{linear-interval order} 
(also known as a \emph{PI order}~\cite{COS08-ENDM}) 
if there is a pair of a linear order $L$ and 
an interval order $P_I$ such that  $P = L \cap P_I$. 
Equivalently, a partial order $P = (V, \prec_P)$ is a linear-interval order 
if for each element $v \in V$, there is a triangle $T(v)$ defined by 
a point on the top line $L_1$ and an interval on the bottom line $L_2$
(recall that $L_1$ and $L_2$ are two horizontal lines with $L_1$ above $L_2$) 
such that $u \prec_P v$ if and only if 
$T(u)$ lies completely to the left of $T(v)$. 
See Figures~\ref{fig:example}\subref{fig:C-representation} 
and~\ref{fig:example}\subref{fig:C-order} for example. 

A linear order $L = (V, \prec_L)$ is called a \emph{linear extension} of 
a partial order $P = (V, \prec_P)$ if $u \prec_L v$ whenever $u \prec_P v$. 
Hence, the linear extension $L$ of $P$ has all the relations of $P$ with 
the additional relations that make $L$ linear. 
We define two properties of linear extensions. 
\begin{itemize}
\item 
Let $\mathbf{2+2}$ denote the partial order 
consisting of four elements $a_0, a_1, b_0, b_1$ 
whose only relations are $a_0 \prec_P b_0$ and $a_1 \prec_P b_1$. 
A linear extension $L = (V, \prec_L)$ of $P = (V, \prec_P)$ 
is said to fulfill the \emph{$\mathbf{2+2}$ rule} if 
for every suborder $\mathbf{2+2}$ in $P$, 
either $b_0 \prec_L a_1$ or $b_1 \prec_L a_0$. 

\item 
An \emph{alternating $2k$-anticycle} of a linear extension $L = (V, \prec_L)$ 
of $P = (V, \prec_P)$ is an induced suborder consisting of distinct $2k$ elements 
$a_0, b_0, a_1, b_1, \ldots, a_{k-1}, b_{k-1}$ with 
$a_i \prec_P b_i$ and $a_{i+1} \prec_L b_i$ but $a_{i+1} \not\prec_P b_i$ 
for any $i = 0, 1, \ldots, k-1$ (indices are modulo $k$). 
See Figure~\ref{fig:anticycle} for example. 
\end{itemize}
Notice that a linear extension $L$ of $P$ fulfills the $\mathbf{2+2}$ rule 
if and only if $L$ contains no alternating 4-anticycle. 
These properties characterize the linear-interval orders as follows. 
\begin{theorem}\label{theorem:partial order}
For a partial order $P$, the following conditions are equivalent: 
\begin{enumerate}[label=\upshape{(\roman*)}]
\item $P$ is a linear-interval order; 
\item $P$ has a linear extension fulfilling the $\mathbf{2+2}$ rule; 
\item $P$ has a linear extension that contains no alternating 4-anticycle. 
\end{enumerate}
\end{theorem}
\begin{proof}
It is obvious that 
\textbf{\textrm{(ii)}} $\iff$ \textbf{\textrm{(iii)}}. 
The implications 
\textbf{\textrm{(i)}} $\Longrightarrow$ \textbf{\textrm{(ii)}} and 
\textbf{\textrm{(iii)}} $\Longrightarrow$ \textbf{\textrm{(i)}} 
are proved by Lemma~\ref{lemma:2+2} and~\ref{lemma:proof}, 
respectively.
\end{proof}

\begin{figure*}[t]
  \psfrag{a1}{$a_0$}
  \psfrag{a2}{$a_1$}
  \psfrag{a3}{$a_2$}
  \psfrag{b1}{$b_0$}
  \psfrag{b2}{$b_1$}
  \psfrag{b3}{$b_2$}
  \centering\subcaptionbox{The alternating 4-anticycle. \label{fig:4-anticycle}}
  {\includegraphics[scale=0.6]{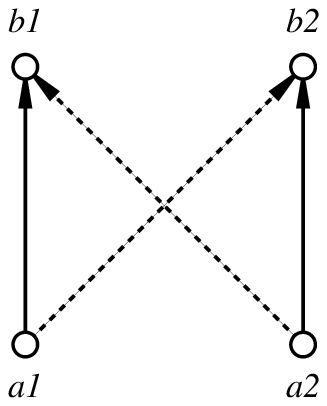}}
  \centering\subcaptionbox{The alternating 6-anticycle. \label{fig:6-anticycle}}
  {\includegraphics[scale=0.6]{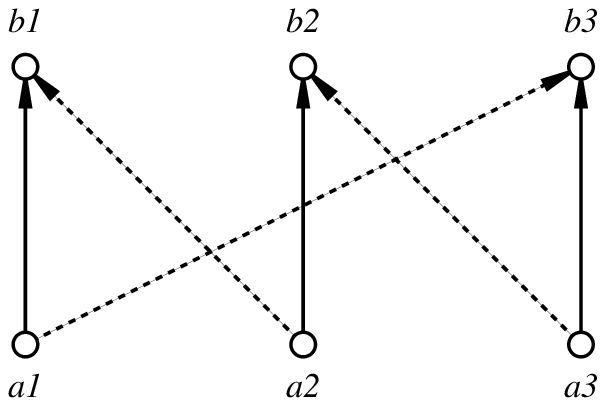}}
  \caption{
    Alternating anticycles. 
    An arrow $a \to b$ denotes the relation $a \prec_P b$, and 
    a dashed arrow $a \dasharrow b$ denotes the relation $a \prec_L b$ 
    but $a \not\prec_P b$. 
    }
  \label{fig:anticycle}
\end{figure*}

\begin{lemma}\label{lemma:2+2}
If a partial order $P = (V, \prec_P)$ has a pair of 
a linear order $L = (V, \prec_L)$ and an interval order $P_I = (V, \prec_I)$ 
with $P = L \cap P_I$, then for any suborder $\mathbf{2+2}$ in $P$, 
the implications $a_0 \prec_L a_1 \iff b_0 \prec_L a_1 \iff b_0 \prec_L b_1 \iff a_0 \prec_L b_1$ holds. 
\end{lemma}
\begin{proof}
For an element $v$ of $P_I$, let $I(v)$ denote the interval of $v$ 
in the representation of $P_I$. 
We first show that $a_0 \prec_L a_1 \Longrightarrow b_0 \prec_L a_1$. 
Suppose for a contradiction that $a_0 \prec_L a_1 \prec_L b_0$. 
Since $a_0 \prec_P b_0$, 
the interval $I(a_0)$ lies completely to the left of $I(b_0)$. 
Then since $a_0 \prec_L a_1 \prec_L b_0$, 
the interval $I(a_1)$ must intersect both $I(a_0)$ and $I(b_0)$. 
Since $a_1 \prec_P b_1$, 
the interval $I(a_1)$ lies completely to the left of $I(b_1)$, 
and consequently, $I(a_0)$ lies completely to the left of $I(b_1)$.
From $a_1 \prec_P b_1$, we also have $a_0 \prec_L a_1 \prec_L b_1$, 
which implies $a_0 \prec_P b_1$, a contradiction. 
Thus, $a_0 \prec_L a_1 \Longrightarrow b_0 \prec_L a_1$. 
By the similar argument, we have the other implications 
$b_0 \prec_L a_1 \Longrightarrow 
b_0 \prec_L b_1 \Longrightarrow 
a_0 \prec_L b_1 \Longrightarrow a_0 \prec_L a_1$. 
We note that this proof is implicit in~\cite{CK87-CN}. 
\end{proof}

\begin{lemma}\label{lemma:proof}
If a partial order $P = (V, \prec_P)$ has a linear extension $L = (V, \prec_L)$ 
that contains no alternating 4-anticycle, then there is 
an interval order $P_I = (V, \prec_I)$ with $P = L \cap P_I$. 
\end{lemma}
\begin{proof}
We prove the lemma by showing an algorithm to construct 
an interval representation of $P_I$ from $P$ and $L$. 
We note that this algorithm is inspired by the algorithms that solve 
the sandwich problems for chain graphs and for threshold graphs~\cite{DFGKM11-AOR,GKS95-JAL,MP95-book,RS95-STOC,Takaoka16-arXiv}. 
In this proof, we use an arrow $a \to b$ 
to denote the relation $a \prec_P b$, and 
we use a dashed arrow $a \dasharrow b$ to denote 
the relation $a \prec_L b$ but $a \not\prec_P b$ 
as in Figure~\ref{fig:anticycle}. 
Notice that for a partial order $Q$, 
the intersection $L \cap Q = P$ if and only if 
$Q$ has all the relations of $\to$ 
but has no relations of $\dasharrow$. 
The following facts are central to the proof 
of the correctness of the algorithm. 

\begin{claim}
$L$ contains no alternating anticycle. 
\end{claim}
\begin{proof}
Suppose for a contradiction that $L$ contains an alternating anticycle. 
Let $C$ be an alternating $2k$-anticycle of $L$ with the least number of elements, 
and let $a_0, b_0, a_1, b_1, \ldots, a_{k-1}, b_{k-1}$ be 
the consecutive elements of $C$ with 
$a_i \to b_i$ and $a_{i+1} \dasharrow b_i$ 
for any $i = 0, 1, \ldots, k-1$ (indices are modulo $k$). 
Since $L$ contains no alternating 4-anticycle, we have $k \geq 3$. 
We consider the relation between $a_0$ and $b_1$. 
If $a_0 \to b_1$ 
then the elements $a_0, b_1, a_2, b_2, a_3, b_3, \ldots, a_{k-1}, b_{k-1}$ 
induce an alternating $(2k-2)$-anticycle, 
contradicting the minimality of $C$. 
If $b_1 \to a_0$ then $a_1 \to b_0$ 
by the transitivity of $\prec_P$, a contradiction. 
If $a_0 \dasharrow b_1$ 
then the elements $a_0, a_1, b_0, b_1$ induce an alternating 4-anticycle, 
a contradiction. 
Therefore, we have $b_1 \dasharrow a_0$. 
Similarly, we have $b_{i+1} \dasharrow a_i$ for any $i = 0, 1, \ldots, k-1$. 
However, it follows from $a_i \to b_i$ that 
$L$ is not a linear order, a contradiction. 
\end{proof}

An element $a$ of $P$ is said to be \emph{minimal} 
if there is no element $b$ of $P$ with $b \prec_P a$. 
Let $S$ be the set of all minimal elements of $P$. 
\begin{claim}\label{claim:2}
There is a minimal element $a \in S$ such that 
for any element $b \in V \setminus S$, if $a \prec_L b$ then $a \prec_P b$. 
In other words, there is an element $a \in S$ that 
has no element $b \in V \setminus S$ with $a \dasharrow b$. 
\end{claim}
\begin{proof}
Suppose for a contradiction that 
for any minimal element $a \in S$, 
there is an element $b \in V \setminus S$ with $a \dasharrow b$. 
Notice that for any element $b \in V \setminus S$, 
there is a minimal element $a \in S$ with $a \to b$. 
Thus, we can grow a path alternating between $S$ and $V \setminus S$ 
until an alternating anticycle is obtained, 
contradicting that $L$ contains no alternating anticycle. 
\end{proof}

\begin{algorithm}
\caption{Constructing of the interval representations}
\label{algorithm:construction}
  \KwData{The partial order $P = (V, \prec_P)$ and 
    the linear extension $L = (V, \prec_L)$ of $P$}
  \KwResult{An interval representation $\{I(v) = [l(v), r(v)] \mid v \in V\}$ of $P_I = (V, \prec_I)$ with $P = L \cap P_I$}
  $S \gets \emptyset$, $i \gets 0$\;
  \Repeat{$V = \emptyset$}{
    $i \gets i + 1$\;
    \ForEach{element $a \in V \setminus S$}{
      \If{$a$ has no element $b \in S$ with $b \prec_P a$}{
        $S \gets S \cup \{a\}$\;
        $l(a) \gets i$\;
      }
    }
    $i \gets i + 1$\;
    \ForEach{element $a \in S$}{
      \If{$a$ has no element $b \in V \setminus S$ with $a \prec_L b$ but $a \not\prec_P b$}{
        $V \gets V \setminus \{a\}$, $S \gets S \setminus \{a\}$\;
        $r(a) \gets i$\;
      }
      \tcc{Claim~\ref{claim:2} ensures that at least one element of $S$ fulfills the {\sf\bf if} condition. }
    }
  }
\end{algorithm}

The algorithm to construct an interval representation of $P_I$ 
is given as Algorithm~\ref{algorithm:construction}. 
In the end of the loop at Line~4, 
$S$ has all the minimal elements of the suborder of $P$ induced by $V$ 
(recall that elements may be removed from $V$ at Line~13). 
Hence, Claim~\ref{claim:2} ensures that 
$S$ has at least one element fulfilling 
the {\sf\bf if} condition at Line~12. 
Since such an element is removed from $V$ and $S$ at Line~13, 
we can see by induction that 
Algorithm~\ref{algorithm:construction} eventually terminate. 
For any two elements $a, b \in V$, 
the {\sf\bf if} condition at Line~5 ensures that 
$r(a) < l(b)$ whenever $a \to b$, and 
the {\sf\bf if} condition at Line~12 ensures that 
$a \not\dasharrow b$ whenever $r(a) < l(b)$; 
the interval order $P_I$ has all the relations of $\to$ 
but has no relations of $\dasharrow$. 
Hence, Algorithm~\ref{algorithm:construction} gives 
an interval representation of $P_I$ with $P = L \cap P_I$, 
and we have Lemma~\ref{lemma:proof}. 
\end{proof}

\section{Apex orderings}\label{section:apex}
An \emph{orientation} of a graph $G$ is an assignment of 
a direction to each edge of $G$. 
A \emph{transitive orientation} of $G$ 
is an orientation such that if 
for any three vertices $u, v, w$ of $G$, 
$u \to v$ and $v \to w$ then $u \to w$. 
A transitively oriented graph is used to 
represent a partial order $P = (V, \prec_P)$, 
where an edge $u \to v$ denotes the relation $u \prec_P v$. 
A graph is called a \emph{comparability graph} 
if it has a transitive orientation. 
For a graph $G = (V, E)$, the \emph{complement} of $G$ 
is the graph $\overline{G} = (V, \overline{E})$ 
such that for any two vertices $u, v \in V$, 
$uv \in \overline{E} \iff uv \notin E$. 
The complement of a comparability graph 
is called a \emph{cocomparability graph}. 
The vertex ordering characterizations of 
these graph classes are known as follows~\cite{BLS99,KS93-SIAMDM}. 
Here, if $\sigma$ is a vertex ordering of $G$, 
we use $u <_{\sigma} v$ to denote that 
$u$ precedes $v$ in $\sigma$. 
\begin{itemize}
\item
A graph $G = (V, E)$ is a comparability graph if and only if 
there is a vertex ordering $\sigma$ of $G$ such that 
for any three vertices $u <_{\sigma} v <_{\sigma} w$, 
if $uv \in E$ and $vw \in E$ then $uw \in E$. 
We call such an ordering a \emph{comparability ordering}. 
In other words, a vertex ordering $\sigma$ is a comparability ordering 
if and only if $\sigma$ contains no subordering 
in Figure~\ref{fig:forbidden-patterns}\subref{fig:cp}. 
\item
A graph $G = (V, E)$ is a cocomparability graph if and only if 
there is a vertex ordering $\sigma$ of $G$ such that 
for any three vertices $u <_{\sigma} v <_{\sigma} w$, 
if $uw \in E$ then $uv \in E$ or $vw \in E$. 
We call such an ordering a \emph{cocomparability ordering}. 
In other words, a vertex ordering $\sigma$ is a cocomparability ordering 
if and only if $\sigma$ contains no subordering 
in Figure~\ref{fig:forbidden-patterns}\subref{fig:cpc}. 
\end{itemize}

Simple-triangle graphs are characterized by the following vertex ordering properties. 
\begin{itemize}
\item 
Let $C_4 = (u, v, w, x)$ denote a chordless cycle of length 4. 
A vertex ordering $\sigma$ of $G$ is said to fulfill the \emph{$C_4$ rule} 
if for every cycle $C_4$ in $G$, 
the implications $u <_{\sigma} v \iff w <_{\sigma} v \iff w <_{\sigma} x \iff u <_{\sigma} x$ holds. 
\item 
Let $2K_2$ denote the graph 
consisting of four vertices $u, v, w, x$ 
whose only edges are $uw$ and $vx$. 
A vertex ordering $\sigma$ of $G$ is said to fulfill the \emph{$2K_2$ rule} 
if for every subgraph $2K_2$ in $G$, 
the implications $u <_{\sigma} v \iff w <_{\sigma} v \iff w <_{\sigma} x \iff u <_{\sigma} x$ holds. 
We note that the $2K_2$ rule are also used to characterize 
co-threshold tolerance graphs~\cite{BLS99,MRT88-JGT}. 
\end{itemize}
Notice that the ordering is 
a vertex ordering of $G$ fulfilling the $C_4$ rule if and only if 
it is a vertex ordering of the complement $\overline{G}$ of $G$ 
fulfilling of the $2K_2$ rule. 
These rules characterize the simple-triangle graphs as follows. 
\begin{theorem}\label{theorem:C4}
For a graph $G$, the following conditions are equivalent: 
\begin{enumerate}[label=\upshape{(\roman*)}]
\item $G$ is a simple-triangle graph; 
\item $G$ has a cocomparability ordering fulfilling the $C_4$ rule; 
\item $G$ has a vertex ordering that contains no subordering 
in Figures~\ref{fig:forbidden-patterns}\subref{fig:cpc},~\ref{fig:forbidden-patterns}\subref{fig:p1}, 
and~\ref{fig:forbidden-patterns}\subref{fig:p2}; 
\item $\overline{G}$ has a comparability ordering fulfilling the $2K_2$ rule; 
\item $\overline{G}$ has a vertex ordering that contains no subordering 
in Figures~\ref{fig:forbidden-patterns}\subref{fig:cp},~\ref{fig:forbidden-patterns}\subref{fig:p1}, 
and~\ref{fig:forbidden-patterns}\subref{fig:p2}. 
\end{enumerate}
\end{theorem}

\begin{figure*}[t]
  \psfrag{u}{$u$}
  \psfrag{v}{$v$}
  \psfrag{w}{$w$}
  \psfrag{x}{$x$}
  \centering\subcaptionbox{\label{fig:cp}}
  {\includegraphics[scale=0.7]{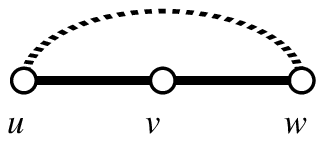}}
  \centering\subcaptionbox{\label{fig:cpc}}
  {\includegraphics[scale=0.7]{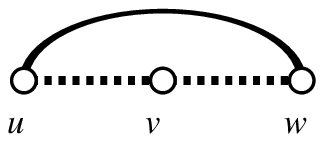}}
  \centering\subcaptionbox{\label{fig:p1}}
  {\includegraphics[scale=0.7]{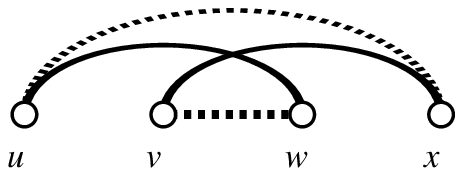}}
  \centering\subcaptionbox{\label{fig:p2}}
  {\includegraphics[scale=0.7]{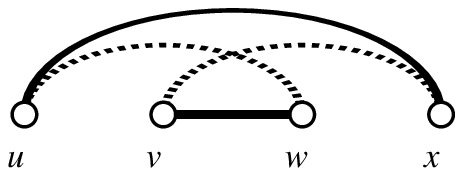}}
  \caption{Forbidden patterns. 
    Lines and dashed lines denote edges and non-edges, respectively. 
    Edges that may or may not be present is not drawn. 
  }
  \label{fig:forbidden-patterns}
\end{figure*}

\begin{proof}
It is obvious that 
\textbf{\textrm{(ii)}} $\iff$ \textbf{\textrm{(iv)}} and 
\textbf{\textrm{(iii)}} $\iff$ \textbf{\textrm{(v)}}. \par
\textbf{\textrm{(i)} $\Longrightarrow$ \textrm{(ii)}:}
It suffices to show that the apex ordering $\sigma$ 
of a simple-triangle graph $G = (V, E)$ is a cocomparability ordering 
since an apex ordering is known to fulfill the $C_4$ rule~\cite{CK87-CN} 
(see also Lemma~\ref{lemma:2+2}). 
Suppose that $G$ has three vertices $u <_{\sigma} v <_{\sigma} w$ with $uw \in E$. 
Here, we use $T(v)$ to denote the triangle of a vertex $v$ in the representation of $G$. 
Since $u <_{\sigma} v <_{\sigma} w$ and $T(u) \cap T(w) \neq \emptyset$, 
the triangle $T(v)$ must intersect $T(u)$ or $T(w)$.
Hence, we have $uv \in E$ or $vw \in E$. 
\par
\textbf{\textrm{(iv)} $\Longrightarrow$ \textrm{(i)}:}
Let $\sigma$ be a comparability ordering 
of $\overline{G} = (V, \overline{E})$ fulfilling the $2K_2$ rule. 
From $\overline{G}$, we can obtain the partial order $P$ 
if we orient the edges of $\overline{G}$ transitively 
so that $u \to v \iff u <_{\sigma} v$ 
since $\sigma$ is a comparability ordering. 
Since $\sigma$ fulfills the $2K_2$ rule, 
$\sigma$ is also a linear extension of $P$ fulfilling the $\mathbf{2+2}$ rule. 
By Theorem~\ref{theorem:partial order}, $P$ is a linear-interval order, 
and hence, $G$ is a simple-triangle graph. 
\par
\textbf{\textrm{(ii)} $\Longrightarrow$ \textrm{(iii)}:}
Let $\sigma$ be a cocomparability ordering 
of $G = (V, E)$ fulfilling the $C_4$ rule. 
The ordering $\sigma$ contains no subordering 
in Figure~\ref{fig:forbidden-patterns}\subref{fig:cpc} 
since $\sigma$ is a cocomparability ordering. 
Suppose for a contradiction that 
there are four vertices $u <_{\sigma} v <_{\sigma} w <_{\sigma} x$ on $\sigma$ 
that induce a subordering in Figure~\ref{fig:forbidden-patterns}\subref{fig:p1}. 
We have $uv \in E$ for otherwise the vertices 
$u <_{\sigma} v <_{\sigma} w$ would induce 
a subordering in Figure~\ref{fig:forbidden-patterns}\subref{fig:cpc}. 
We also have $wx \in E$ for otherwise the vertices 
$v <_{\sigma} w <_{\sigma} x$ would induce 
a subordering in Figure~\ref{fig:forbidden-patterns}\subref{fig:cpc}. 
Hence, the vertices $u <_{\sigma} v <_{\sigma} w <_{\sigma} x$ 
induce $C_4$ that violates the $C_4$ rule, a contradiction. 
Similarily, suppose for a contradiction that 
there are four vertices $u <_{\sigma} v <_{\sigma} w <_{\sigma} x$ on $\sigma$ 
that induce a subordering in Figure~\ref{fig:forbidden-patterns}\subref{fig:p2}. 
We have $uv \in E$ for otherwise the vertices 
$u <_{\sigma} v <_{\sigma} x$ would induce 
a subordering in Figure~\ref{fig:forbidden-patterns}\subref{fig:cpc}. 
We also have $wx \in E$ for otherwise the vertices 
$u <_{\sigma} w <_{\sigma} x$ would induce 
a subordering in Figure~\ref{fig:forbidden-patterns}\subref{fig:cpc}. 
Hence, the vertices $u <_{\sigma} v <_{\sigma} w <_{\sigma} x$ 
induce $C_4$ that violates the $C_4$ rule, a contradiction. 
\par
\textbf{\textrm{(ii)} $\Longleftarrow$ \textrm{(iii)}:}
Let $\sigma$ be a vertex ordering that contains no subordering 
in Figures~\ref{fig:forbidden-patterns}\subref{fig:cpc},~\ref{fig:forbidden-patterns}\subref{fig:p1}, 
and~\ref{fig:forbidden-patterns}\subref{fig:p2}. 
The ordering $\sigma$ is a cocomparability ordering 
since $\sigma$ contains no subordering 
in Figure~\ref{fig:forbidden-patterns}\subref{fig:cpc}. 
We can verify that any four vertices of $C_4$ 
that violates the $C_4$ rule induce the subordering in either 
Figure~\ref{fig:forbidden-patterns}\subref{fig:p1} or 
~\ref{fig:forbidden-patterns}\subref{fig:p2}. 
Hence, $\sigma$ fulfills the $C_4$ rule. 
\par
We can also prove \textbf{\textrm{(iv)} $\iff$ \textrm{(v)}} 
by the similar argument in the proof 
of \textbf{\textrm{(ii)} $\iff$ \textrm{(iii)}}. 
\end{proof}

We can also describe the characterization in terms of orientations of graphs. 
An orientation of a graph is called \emph{acyclic} if it has no directed cycle. 
An orientation of a graph is called \emph{alternating} 
if it is transitive on every chordless cycle of length 
greater than or equal to 4, that is, 
the directions of the oriented edges alternate. 
A graph is called \emph{alternately orientable}~\cite{Hoang87-JCTSB} 
if it has an alternating orientation. 
%
%
Since cocomparability graphs has no chordless cycle of length greater than 4, 
we have the following from Theorem~\ref{theorem:C4}. 
\begin{corollary}\label{corollary:alternating orientation}
A graph $G$ is a simple-triangle graph if and only if there is 
an alternating orientation of $G$ and 
a transitive orientation of the complement $\overline{G}$ of $G$ 
such that the union of the oriented edges of $G$ and $\overline{G}$ form 
an acyclic orientation of the complete graph. 
\end{corollary}
Moreover, we have the following from Theorem~\ref{theorem:partial order} 
since being a linear-interval order 
is a comparability invariant~\cite{COS08-ENDM}. 
\begin{corollary}
Let $G$ be a simple-triangle graph. 
For any transitive orientation of the complement $\overline{G}$ of $G$, 
there is an alternating orientation of $G$ 
such that the union of the oriented edges of $G$ and $\overline{G}$ form 
an acyclic orientation of the complete graph. 
\end{corollary}

\section{Concluding remarks}\label{section:conclusion}
We have shown a vertex ordering characterization of simple-triangle graphs 
based on the ordering of the apices of the triangles in the representation. 
We conclude this paper with some miscellaneous topics related to this characterization. 

\begin{figure*}[t]
  \psfrag{a1}{$a_1$}
  \psfrag{a2}{$a_2$}
  \psfrag{a3}{$a_3$}
  \psfrag{b1}{$b_1$}
  \psfrag{b2}{$b_2$}
  \psfrag{b3}{$b_3$}
  \psfrag{c1}{$c_1$}
  \psfrag{c2}{$c_2$}
  \psfrag{c3}{$c_3$}
  \psfrag{c4}{$c_4$}
  \centering\subcaptionbox{The order $W$. \label{fig:W-order}}
  {\includegraphics[scale=0.6]{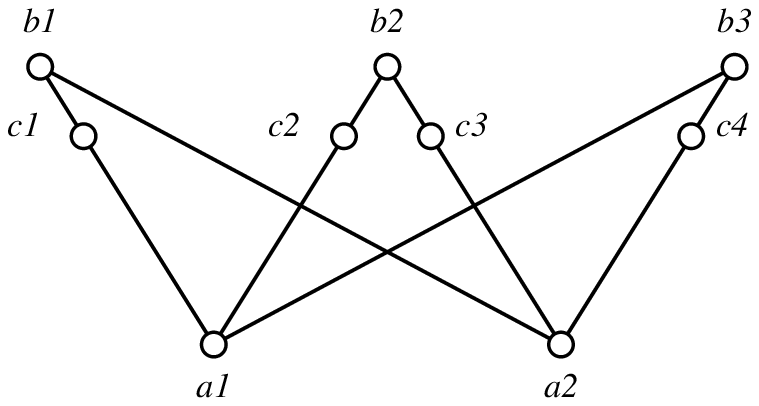}}
  \centering\subcaptionbox{The graph $\overline{W}$. \label{fig:W-graph}}
  {\includegraphics[scale=0.6]{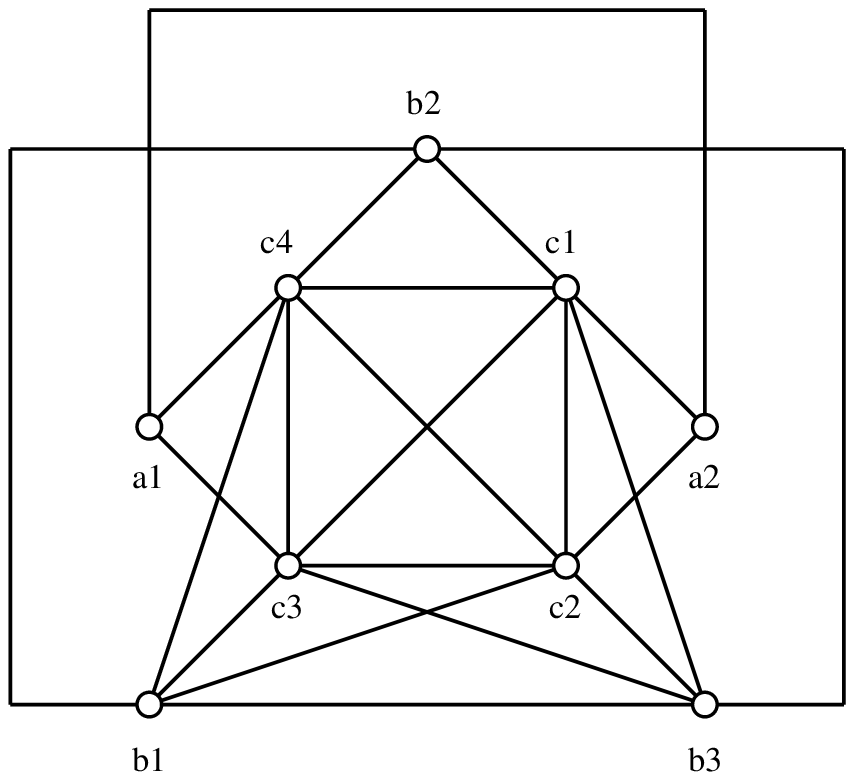}}\\
  \centering\subcaptionbox{The order $IV$. \label{fig:IV-order}}
  {\includegraphics[scale=0.6]{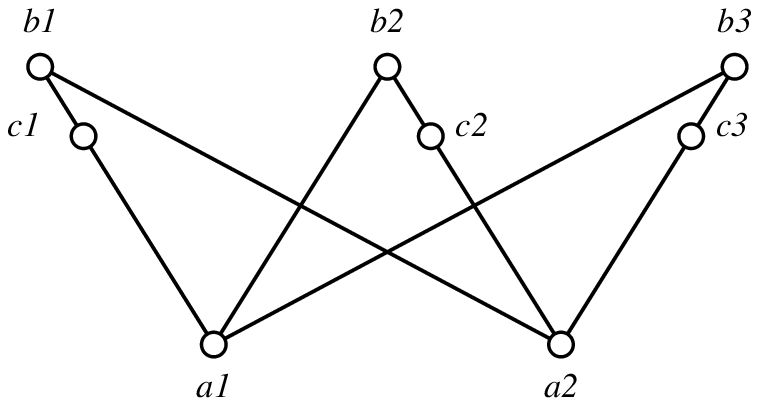}}
  \centering\subcaptionbox{The graph $\overline{IV}$. \label{fig:IV-graph}}
  {\includegraphics[scale=0.6]{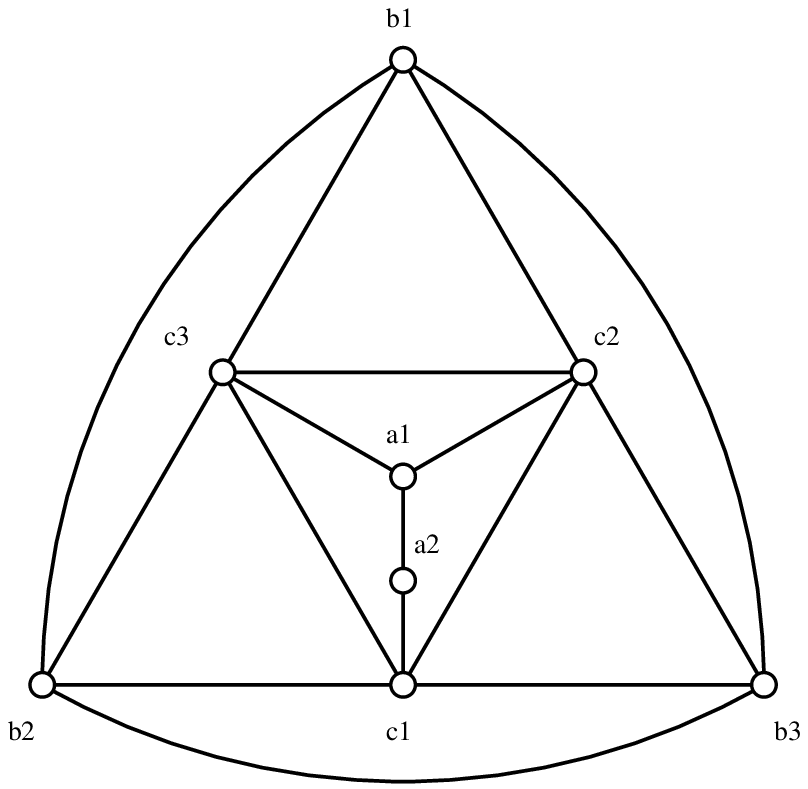}}
  \caption{The partial orders and the graphs. }
  \label{fig:conclusion}
\end{figure*}

Corollary~\ref{corollary:alternating orientation} indicates that 
a simple-triangle graph is a cocomparability graph that has 
an alternating orientation~\cite{Felsner98-JGT}, 
but we can see the converse is not true. 
The separating example is the graph $\overline{W}$ 
in Figure~\ref{fig:conclusion}\subref{fig:W-graph}. 
This graph $\overline{W}$ has the unique alternating orientation (up to reversal), 
and the complement of $\overline{W}$ 
has the unique transitive orientation (up to reversal) 
whose Hasse diagram is shown in Figure~\ref{fig:conclusion}\subref{fig:W-order}. 
Suppose that $a_1 \to a_2$. 
Then the cycle $(a_1, a_2, c_2, c_3)$ requires that 
$a_1 \to a_2 \iff c_2 \to a_2$, 
while the cycles 
$(a_1, a_2, c_1, c_3)$, 
$(c_1, c_3, b_1, b_2)$, and 
$(c_1, c_2, b_1, b_2)$ require that 
$a_1 \to a_2 \iff c_1 \to c_3 \iff 
b_1 \to b_2 \iff b_1 \to c_2$. 
Hence, we have a directed cycle $(b_1, c_2, a_2)$ 
in the union of the oriented edges of $G$ and $\overline{G}$. 
Suppose on the contrary that $a_1 \gets a_2$. 
Then the cycle $(a_1, a_2, c_2, c_3)$ requires that 
$a_1 \gets a_2 \iff a_1 \gets c_3$, 
while the cycles 
$(a_1, a_2, c_2, c_4)$, 
$(c_2, c_4, b_2, b_3)$, and 
$(c_3, c_4, b_2, b_3)$ require that 
$a_1 \gets a_2 \iff c_2 \gets c_4 \iff 
b_2 \gets b_3 \iff c_3 \gets b_3$. 
Hence, we have a directed cycle $(b_3, c_3, a_1)$ 
in the union of the oriented edges of $G$ and $\overline{G}$, and 
Corollary~\ref{corollary:alternating orientation} indicates that 
$\overline{W}$ is not a simple-triangle graph. 

A graph is a \emph{permutation graph} if 
it is simultaneously a comparability graph and a cocomparability graph. 
A permutation graph $G$ is known to have the unique transitive orientation 
(up to reversal) when the complement $\overline{G}$ of $G$ 
has the unique transitive orientation (see~\cite{Golumbic04} for example). 
This derives the polynomial-time algorithm for testing isomorphism of 
permutation graphs~\cite{Colbourn81-Networks}. 
Hence, it is natural to ask whether a simple-triangle graph $G$ has 
the unique alternating orientation when the complement $\overline{G}$ of $G$ 
has the unique transitive orientation (up to reversal). 
We give the negative answer to this question. 
The graph $\overline{IV}$ in Figure~\ref{fig:conclusion}\subref{fig:IV-graph} 
does not have the unique alternating orientation since we can reverse 
the orientation of edges on the cycle $(b_2, b_3, c_2, c_3)$, 
while the complement of $\overline{IV}$ 
has the unique transitive orientation (up to reversal) 
whose Hasse diagram is shown 
in Figure~\ref{fig:conclusion}\subref{fig:IV-order}. 

We finally pose two open questions for simple-triangle graphs. 
The first question is related to the recognition problem. 
The polynomial-time recognition algorithm is already 
known~\cite{Mertzios15-SIAMDM,Takaoka16-arXiv}, 
but the running time of it is $O(n^2\bar{m})$, 
where $n$ and $\bar{m}$ is the number of vertices and non-edges 
of the graph, respectively. 
The algorithm reduces the recognition to a problem of 
covering an associated bipartite graph by two chain graphs with additional conditions. 
Our first question is that 
can we recognize simple-triangle graphs in polynomial time 
by using the vertex ordering characterization in this paper? 
\begin{problem}
By using the vertex ordering characterization of simple-triangle graphs, 
find a recognition algorithm faster than 
the existing ones~\cite{Mertzios15-SIAMDM,Takaoka16-arXiv}. 
\end{problem}
The second question is related to the isomorphism problem. 
A \emph{canonical ordering} of a graph $G$ is a vertex ordering of $G$ 
such that every graph that is isomorphic to $G$ 
has the same canonical ordering as $G$. 
Hence, the graph isomorphism problem can be solved by 
computing the canonical orderings of the two given graphs and 
testing whether these two ordered graphs are identical. 
Our second question is that is there any canonical ordering of simple-triangle graphs 
based on the vertex ordering characterization in this paper? 
\begin{problem}
By using the vertex ordering characterization of simple-triangle graphs, 
find a canonical ordering computable in polynomial time. 
\end{problem}



\end{document}